\documentclass[10pt,a4paper,oneside]{article}
\usepackage{enumitem}
\usepackage[utf8]{inputenc}
\usepackage[T1]{fontenc}
\usepackage{mathrsfs, amsmath}
\usepackage{amsfonts}
\usepackage{amssymb}
\usepackage{ifthen}
\usepackage[margin=1in]{geometry}
\usepackage{graphicx}

 \usepackage{amsmath, xspace}
\usepackage{epsfig}
\usepackage{longtable}
\usepackage{color}
\usepackage{subfigure}

\def\sthat{\ \ \mbox{such that}\ \ }
\usepackage{courier}

\newtheorem{theorem}{Theorem}
\newtheorem{lemma}[theorem]{Lemma}

\def\bkE{{\rm I\kern-.17em E}}
\def\bk1{{\rm 1\kern-.17em l}}
\def\bkD{{\rm I\kern-.17em D}}
\def\bkR{{\rm I\kern-.17em R}}
\def\bkP{{\rm I\kern-.17em P}}

\def\bkZ{{\bf{Z}}}

\def\bfone{{\bf 1}}

\def\bkE{{\rm I\kern-.17em E}}
\def\bk1{{\rm 1\kern-.17em l}}
\def\bkD{{\rm I\kern-.17em D}}
\def\bkR{{\rm I\kern-.17em R}}
\def\bkP{{\rm I\kern-.17em P}}

\def\bkZ{{\bf{Z}}}
\def\b12{(\beta_1,\beta_2)}

\newenvironment{proof}[1][]{{\noindent \bf Proof #1: }}{\hfill \qed\vspace{6pt}}

\newcounter{example}
\renewcommand{\theexample}{\arabic{example}}
\newenvironment{examplec}[1][]{\refstepcounter{example}
\par\medskip \noindent%
   \textbf{Example~\theexample. #1} \rmfamily}{\hfill $\square$   \hspace{-4.5pt} \vspace{6pt}}

\newlength{\noteWidth}
\long\def\notes#1{\ifinner
{\tiny #1}
\else
\marginpar{\parbox[t]{\noteWidth}{\raggedright\tiny #1}}
\fi\typeout{#1}}

 \def\notes#1{\typeout{read notes: #1}} 

\newcommand{\ie}{i.e.\@\xspace} 
\newcommand{\eg}{e.g.\@\xspace} 

\newcommand{\Real}{\mathbb{R}}

\newcommand{\minimize}[1]{\displaystyle\minim_{#1}}
\newcommand{\minim}{\mathop{\hbox{\rm minimize}}}
\newcommand{\maximize}[1]{\displaystyle\maxim_{#1}}
\newcommand{\maxim}{\mathop{\hbox{\rm maximize}}}

\def\subject{\hbox{\rm subject to}}

\def\diag{\mathop{\hbox{\rm diag}}}

\def\half  {{\textstyle{1\over 2}}}

\def\conv{\textrm{conv}\:}

\def\norm#1{\|#1\|}
\def\spose#1{\hbox to 0pt{#1\hss}}

\def\text #1{\hbox{\quad#1\quad}}

\def\nthinsp{\mskip -2   mu}

\def\superstar{^{\raise 0.5pt\hbox{$\nthinsp *$}}}
\def\SUPERSTAR{^{\raise 0.5pt\hbox{$*$}}}

\def\lamstarT {\lambda^{\raise 0.5pt\hbox{$\nthinsp *$}T}}

\def\aur{\;\textrm{and}\;}

\def\SOL{{\rm SOL}}

\def\LCP{{\rm LCP}}

\let\forallnew\forall
\renewcommand{\forall}{\forallnew\ }
\let\forall\forallnew

\def\t{^\top}
		\def\bkE{{\rm I\kern-.17em E}}
		\def\bk1{{\rm 1\kern-.17em l}}
		\def\bkD{{\rm I\kern-.17em D}}
		\def\bkR{{\rm I\kern-.17em R}}
		\def\bkP{{\rm I\kern-.17em P}}
		\def\bkY{{\bf \kern-.17em Y}}
		\def\bkZ{{\bf \kern-.17em Z}}
		\def\bkC{{\bf  \kern-.17em C}}

{\begin{list}{}%
         {\setlength{\leftmargin}{#1}}%
         \item[]%
}
{\end{list}}

		\def\bsp{\begin{split}}
		\def\beq{\begin{eqnarray}}
		\def\bal{\begin{align*}}
		\def\bc{\begin{center}}
		\def\be{\begin{enumerate}}
		\def\bi{\begin{itemize}}
		\def\bs{\begin{small}}
		\def\bS{\begin{slide}}
		\def\ec{\end{center}}
		\def\ee{\end{enumerate}}
		\def\ei{\end{itemize}}
		\def\es{\end{small}}
		\def\eS{\end{slide}}
		\def\eeq{\end{eqnarray}}
		\def\eal{\end{align*}}
		\def\esp{\end{split}}
		\def\qed{ \vrule height7.5pt width7.5pt depth0pt}  

		\def\problem#1#2#3#4{\fbox
		 {\begin{tabular*}{0.80\textwidth}
			{@{}l@{\extracolsep{\fill}}l@{\extracolsep{6pt}}l@{\extracolsep{\fill}}c@{}}
				#1 & $\minimize{#2}$ & $#3$ & $ $ \\[5pt]
					 & $\subject\ $    & $#4$ & $ $
			\end{tabular*}}
			}
\def\maxproblem#1#2#3#4{\fbox
		 {\begin{tabular*}{0.80\textwidth}
			{@{}l@{\extracolsep{\fill}}l@{\extracolsep{6pt}}l@{\extracolsep{\fill}}c@{}}
				#1 & $\maximize{#2}$ & $#3$ & $ $ \\[5pt]
					 & $\subject\ $    & $#4$ & $ $
			\end{tabular*}}
			}

				\def\viproblem#1#2#3{\fbox
		 {\begin{tabular*}{0.95\textwidth}
			{@{}l@{\extracolsep{\fill}}l@{\extracolsep{6pt}}l@{\extracolsep{\fill}}c@{}}
				#1 &#2 & $#3 $ 
			\end{tabular*}}}
	\def\cp2problem#1#2#3#4{\fbox
		 {\begin{tabular*}{0.9\textwidth}
			{@{}l@{\extracolsep{\fill}}l@{\extracolsep{6pt}}l@{\extracolsep{\fill}}c@{}}
				#1 & & $#4 $ 
			\end{tabular*}}}

\newcommand{\pmat}[1]{\begin{pmatrix} #1 \end{pmatrix}}
		
		\renewcommand{\emph}[1]{\textbf{#1}}

		\def\bkE{{\rm I\kern-.17em E}}
		\def\bk1{{\rm 1\kern-.17em l}}
		\def\bkD{{\rm I\kern-.17em D}}
		\def\bkR{{\rm I\kern-.17em R}}
		\def\bkP{{\rm I\kern-.17em P}}
		
		\def\bkZ{{\bf{Z}}}

\newcommand {\beeq}[1]{\begin{equation}\label{#1}}
\newcommand {\eeeq}{\end{equation}}
\newcommand {\bea}{\begin{eqnarray}}
\newcommand {\eea}{\end{eqnarray}}

\def\texitem#1{\par\smallskip\noindent\hangindent 25pt
               \hbox to 25pt {\hss #1 ~}\ignorespaces}

\newcommand{\bin}{\{0,1\}}
\newcommand{\scn}[2]{\mathcal{C}_{#1}(#2)}

\newcommand{\bfe}{\textbf{e}}
\newcommand{\STAB}{{\rm STAB}}
\newcommand{\FRAC}{{\rm FRAC}}
\newcommand{\TB}{{\rm TH}}
\title{\bf A linear complementarity based characterization of the weighted independence number and the independent domination number in graphs}
\author{Parthe Pandit \and Ankur A. Kulkarni\thanks{Parthe and Ankur are with the Systems and Control Engineering group at the  
Indian Institute of Technology Bombay
Mumbai, India 400076. They can be contacted at \texttt{\small parthe.pandit@iitb.ac.in} and \texttt{\small kulkarni.ankur@iitb.ac.in}, respectively.} }

\begin{document}
\date{}

\maketitle

\begin{abstract}
The linear complementarity problem is a continuous optimization problem that generalizes convex quadratic programming, Nash equilibria of bimatrix games and several such problems. 
This paper presents a continuous optimization formulation for the weighted independence number of a graph by characterizing it as the maximum weighted $\ell_1$ norm over the solution set of a linear complementarity problem~(LCP). The minimum $\ell_1$ norm of solutions of this LCP is a lower bound on the independent domination number of the graph. Unlike the case of the maximum $\ell_1$ norm, this lower bound is in general weak, but we show it to be tight if the graph is a forest. Using methods from the theory of LCPs, we obtain a few graph theoretic results. In particular, we provide a stronger variant of the Lov\'asz theta of a graph. 
We then provide sufficient conditions for a graph to be well-covered, i.e., for all maximal independent sets to also be maximum. This condition is also shown to be necessary for well-coveredness if the graph is a forest. Finally, the reduction of the maximum independent set problem to a linear program with (linear) complementarity constraints~(LPCC) shows that LPCCs are hard to approximate.
\end{abstract}

\section{Introduction}
This paper concerns a new continuous optimization formulation for the independence number of a graph. An undirected graph $G$ is given by the pair $(V,E)$ where $V$ is a finite set of \textit{vertices} and $E$ is a set of unordered pairs of vertices called \textit{edges}. Two vertices $i,j \in V$ are said to be connected if there exists an edge $(i,j) \in E$ between them. Connected vertices are also called neighbours. An independent set of $G$ is a set of pairwise disconnected vertices and an independent set of largest cardinality called a maximum independent set. The cardinality of the maximum independent set is called the independence number of $G$ denoted $\alpha(G)$. 

Closely related are the concepts of maximality and domination. An independent set is said to be \textit{maximal} if it is not a subset of any larger independent set. Clearly a maximum independent set is maximal but the converse not true in general. A set $S \subseteq V$ is a \textit{dominating set} if every $v \in V\backslash S$ has a neighbour $u\in S.$ One can show that every vertex not in a maximal independent set has at least one neighbour in the set, whereby maximal independent sets are also dominating sets.

Computing the independence number of a general graph is NP-complete, although it is known to be solvable in polynomial time for some subclasses, such as claw-free graphs and perfect graphs \cite{minty1980maximal, grotschel1986relaxations}. Computing the independence number is clearly a discrete optimization problem. However there are several continuous optimization formulations for this quantity. Perhaps the most well known amongst them is the result by Motzkin and Strauss  \cite{motzkin1965maxima} which shows that for a graph $G$ with $n$ vertices,
\[\frac{1}{\alpha(G)} = \min\{x \t (A+I)x \mid \bfe \t x = 1;\; x \geq 0 \},\]
where $\bfe$ is the vector\footnote{Throughout this paper, vectors are column vectors} of 1's in $\Real^n$, $A = [a_{ij}]$ is the adjacency matrix of $G$ (\ie, $a_{ij} = 1$ if $(i,j) \in E$ and $=0$ otherwise), and $I$ is the $n\times n$ identity matrix. Among other continuous formulations, the ones by Harant are noteworthy \cite{harant1999dominating, harant2000some}. Specifically, \cite[Theorem 7]{harant2000some} shows,

\[\alpha(G) = \max\{ \bfe \t x - \frac{1}{2} x \t A x \mid 0 \leq x \leq \bfe \}.\]

For a given weight vector $w\in \Real^n$, the weight of a set $S \subseteq V$ is the quantity $\sum_{i \in S} w_i$. The weighted independence number denoted $\alpha_w(G)$ is the maximum of the weights over all the independent sets, \ie, 

$$\alpha_w(G) := \max \left \{ \sum_{i \in S} w_i \mid S \subseteq V {\rm \;is\; independent}\right \}.$$

Clearly $\alpha(G)$ is $\alpha_{\bfe}(G)$. 
This paper characterizes the weighted independence number of a graph in terms of the linear complementarity problem (LCP). Given a matrix $M \in \Real^{n \times n}$ and $q\in \Real^n$, $\LCP(M,q)$ is the following problem:
\begin{align*}
\text{Find} x=(x_1,x_2\cdots x_n) \in \Real^n \quad \sthat \quad
& (1) \quad x \geq 0,\\
& (2) \quad y = Mx + q \geq 0, \tag*{LCP($M,q$)} \\
& (3) \quad y \t x = 0. 
\end{align*}
Notice that due to the nonnegativity of $x$ and $y$, the last condition is equivalent to requiring $x_iy_i=0$ for all $i.$ This requirement is referred to as the \textit{complementarity condition}. 
A vector $x$ is said to be a solution of $\LCP(M,q)$ if it satisfies the above three conditions. LCPs arise naturally in the characterization of equilibria in bimatrix games and several other problems in operations research. We discuss this problem class later in this paper. 

For a simple graph\footnote{We consider only simple graphs, i.e., graphs without self loops which means $a_{ii} = 0$ for all $i \in V$.} $G=(V,E)$ with $|V|=n$ vertices, consider the $\LCP(A+I,-\bfe)$, i.e., 
$$ \viproblem{$\LCP(G)$}{Find $ x\in \Real^n $ such that }{x\geq 0,\ (A+I)x \geq \bfe, \ x \t \big((A+I)x - \bfe\big)= 0.} $$
We refer to this as $\LCP(G)$ and its solution set as $\SOL(G)$. Let the characteristic vector of a set $S \subseteq V$ be denoted by $\bfone_S$; it is the vector in $\bin^n$ whose $i^{th}$ element is 1 \textit{iff} $i \in S$. It is easy to show that if $S^* \subseteq V$ is a maximum independent set in $G$ then $\bfone_{S^*}$ solves $\LCP(G)$. Consequently, we always have,

\begin{equation} \label{eq:alpha inequality}  \alpha(G) \leq \max\{\bfe\t x \mid x {\rm \;solves\;} \LCP(G)\}.  \end{equation}

Our main result in this paper shows that the inequality in \eqref{eq:alpha inequality} is always tight, even for the weighted independence number.

\begin{theorem}
\label{thm:main} For any simple graph $G$ and weight vector $w \geq 0$,
\[\alpha_w(G) = \max\{w \t x \mid x {\rm \;solves}\; \LCP(A+I,-\bfe)\} ,\]
where $A$ is the adjacency matrix of $G$, $I$ is the identity matrix and $\bfe$ is the vector of 1's in $\Real^n$.
\end{theorem}

To note why the above result is not obvious, consider the quantity $\beta(G)$ defined as the smallest size of a maximal independent set in $G$ (also known as minimum independent dominating set). One can show that the characteristic vector of every maximal independent set solves $\LCP(G)$ (Lemma \ref{lem:integer solutions of lcpg} in the next section). Hence, analogous to \eqref{eq:alpha inequality}, 
\begin{equation}
\beta(G) \geq \min \{\bfe \t x \mid x\ {\rm solves}\ \LCP(A+I,-\bfe)  \}. \label{eq:beta inequality} 
\end{equation}
We show in Section \ref{sec:minimum l1 norm of solg} that this inequality is in general strict, however, equality is achieved when the graph is a forest (\ie, graph that is a union of disjoint trees). Indeed we have the following theorem.

\begin{theorem}
\label{thm:forest min l1 beta}
For a forest $G$,
$$\beta(G) = \min \{\bfe \t x \mid x {\rm \;solves \;} \LCP(A+I,-\bfe)\},$$
where $A$ is the adjacency matrix of $G$, $I$ is the identity matrix and $\bfe$ is the vector of 1's in $\Real^n$.
\end{theorem}

\subsection{Contributions} \label{sec:contributions} 

Our main contribution in this paper are centered around Theorem \ref{thm:main} and its consequences. We also consider the analogous problem of the minimum $\ell_1$ norm of points in $\SOL(G)$ and its relation to $\beta(G).$ As mentioned above, unlike for the \textit{maximum} $\ell_1$ (\ie Theorem~\ref{thm:main}), the inequality in~\eqref{eq:beta inequality} is in general strict; in fact the right hand side in \eqref{eq:beta inequality} need not even be an integer. However, as indicated by Theorem \ref{thm:forest min l1 beta}, we show that this inequality is tight for forest graphs.

We perform a semidefinite programming~(SDP) relaxation for the optimization problem resulting from Theorem \ref{thm:main} to give a tighter version of the Lov\'{a}sz theta~\cite{lovasz1991cones}. The optimization problem in Theorem~\ref{thm:main} results in a more compact integer linear program~(ILP) formulation for $\alpha(G)$ than previous edge-based ILP formulations. The feasible lattice of this new ILP characterizes only the maximal independent sets of the graph. An application of lift-and-project relaxations gives our improved Lov\'{a}sz theta variant. Numerically we have verified that our variant is in general stronger than other Lov\'{a}sz theta variants that employ the same number of constraints.

Graphs for which all maximal independent sets are of the same cardinality are called \textit{well-covered} graphs~\cite{plummer1993well}. Using Theorem~\ref{thm:main} and Theorem~\ref{thm:forest min l1 beta} we derive a new characterization for the well-coveredness for forests: specifically, a forest $G$ is well-covered if and only if $\bfe \t x$ is constant for $x \in \SOL(G)$ (here too the ``if'' direction is easy to see; the ``only if'' needs a proof).

Theorem~\ref{thm:main} gives a characterization of the weighted independence number of a graph via a linear program with (linear) complementarity constraints~(LPCC). An LPCC in its most general form is written as,
\begin{equation*}
\label{eq:LPCC example}
	\maxproblem{LPCC}
	{x,y}
	{c \t x + d \t y}
				 {\begin{array}{r@{\ }c@{\ }l}
					 Bx+Cy & \geq  & b, \\
Mx + Ny +q &\geq & 0, \\
	x &\geq& 0, \\
	x \t (Mx+Ny+q) &=& 0.
	\end{array}}
\end{equation*}
Notice that a feasible pair $(x,y)$ for the LPCC comprises of an  $x$, that solves $\LCP(M,Ny+q)$, and another variable $y$ that parametrizes this LCP, and the pair $(x,y)$ must also satisfy an additional affine constraint $Bx+Cy\geq b.$ Clearly, taking $N,C,B,b,d$ to be 0 vectors or matrices of appropriate dimension gives a special case of the LPCC in which a linear function is maximized over the solution set of an LCP. This precisely the structure of Theorem~\ref{thm:main}.

LPCCs generalize several problem classes including linear programming, and finding sparse (minimum $\ell_0$ norm) solutions of linear equations. Their study is gathering momentum in the operations research literature~\cite{hu2012lpcc, hu2012linear}  as new applications get discovered. Theorem~\ref{thm:main} reveals weighted independence number as another application. LCPs on the other hand are a widely and deeply studied problem class; see, \eg,~\cite{murty1988linear} and~\cite{cottle92linear}. 
Theorem \ref{thm:main} brings in the possibility of using results from the theory of LCPs and LPCCs to develop algorithms or bounds on the independence number.
Indeed our results are obtained by appealing to properties of LCPs.

Finally, the reduction of the independence number problem to an LPCC shows that for an LPCC with $n$ variables, it is NP-hard to approximate it within a factor of $n^{1-\epsilon}$ of its optimal value (assuming $P \neq NP$) even for a strong class of problems with only binary data. This follows from the fact that approximating the independence number is NP-hard due to a result by H{\aa}stad \cite{haastad1996clique}. 

\subsection{Organization of the paper}
The rest of the paper is organized as follows. Section~\ref{sec:prop} elaborates on a few properties of $\LCP(G)$ and recounts some background about LCPs. It is followed by the proof of Theorem \ref{thm:main} in Section~\ref{sec:LPCC formulation}. Section~\ref{sec:applications} derives results pertaining to our SDP relaxation, well-covered graphs and the complexity of LPCCs. The paper concludes in Section~\ref{sec:conc}.

\section{Preliminaries} \label{sec:prop} 

\subsection{Background on LCPs}
Much of what follows is standard and well-documented~\cite{cottle92linear}; we recount it here for the benefit of the reader. 
Linear complementarity problems arise naturally through the modeling of several problems in optimization and allied areas. As an example, consider a convex quadratic program:
\[	\problem{QP}
	{x}
	{\half x\t Q x + c \t x}
				 {\begin{array}{r@{\ }c@{\ }l}
					 Ax& \geq  & b, \quad : \lambda\\
	x &\geq& 0, \\
	\end{array}}\]
	where $Q$ is symmetric and positive semidefinite matrix, and $A$ and $b$ are a matrix and a vector of appropriate dimension. If $\lambda$ denotes the vector of Lagrange multipliers corresponding to the constraint ``$Ax \geq b$", from the Karush-Kuhn-Tucker conditions it is easy to derive that  $x$ solves QP if and only if there exists $\lambda$ such that, 
	\[\pmat{x \\ \lambda} \geq 0,\quad \pmat{Qx +c - A\t \lambda\\ Ax -b} \geq 0,\quad \pmat{x\\ \lambda }\t\pmat{Qx +c - A\t \lambda\\ Ax -b} =0. \]
	This is clearly an LCP in the $(x,\lambda)$-space. 
	
	Another, famous, example comes from Nash equilibria of two person games. Consider a simultaneous move game with two players and loss matrices $A, B \in \Real^{m\times n}$. A Nash equilibrium~\cite{nash51noncooperative} is a pair of vectors $(x^*,y^*) \in \Delta_n \times \Delta_m$ such that, 
	\[(x^*)\t A y^* \leq x \t A y^*, \quad \forall \;x \in \Delta_n, \qquad (x^*)\t B y^* \leq (x^*)\t B y,  \quad \forall \;y \in \Delta_m,\]
	where $\Delta_k$ is the probability simplex in $\Real^k, \Delta_k :=\{x \in \Real^k | \sum_i x_i =1, x\geq 0\}.$ Assuming $A,B$ have positive entries, by suitable transformations~(see, \eg, \cite[p.\ 6]{cottle92linear}), it can be shown that if $(x^*,y^*)$ is a Nash equilibrium, then $(x',y')$, where, 
\[x'=x^*/(x^*)\t By^* \qquad y'=y^*/(x^*)\t A y^*,\]	
	 solves $\LCP(M,q)$ with, 
\[M = \pmat{0 & A \\ B\t &0}, \qquad q = -\bfe,\]
where $\bfe$ denotes a vector of 1's in $\Real^{m+n}.$ Conversely, if $(x',y')$ solves  $\LCP(M,q)$ then $x^*= x'/(\sum_i x'_i)$ and $y^*=y'/\sum_j y'_j$ is a Nash equilibrium. More generally, certain equilibria of games involving coupled constraints~\cite{kulkarni09refinement} also reduce to LCPs. For more applications, we refer the reader to~\cite{cottle92linear}.

LCPs may have unique, finitely many, infinitely many or no solutions. In the case where it has a solution, we say that the LCP is solvable. LCPs with rational inputs are known to be NP-complete~\cite{chung1989np} (Theorem~\ref{thm:main} also yields this as a corollary). 
\textit{Without} the complementarity condition, \ie, requirement (3) in the definition of $\LCP(M,q)$, an LCP  amounts only to finding a feasible point for a set of linear  inequalities. Since the complementarity condition is equivalent to asking that `for all $i$, either $x_i=0$ or $(Mx+q)_i=0$', one can equivalently reformulate the LCP as asking for an $x$ that is feasible for at least one out of $2^n$ systems of linear inequalities. Specifically, if for some subset of indices $S \subseteq \{1,\hdots,n\}$, if $x\in \Real^n$ satisfies the following linear inequalities, 
\[x \geq 0,\ y =Mx+q \geq 0, \quad  x_j =0, \ \forall j\notin S {\rm\;\; and}\ \  y_j=0,\ \ \forall j \in S,\]
then clearly $x$ solves $\LCP(M,q)$. Conversely, if $x$ solves $\LCP(M,q)$ then one may take $S = \{i \mid x_i >0\}$ to verify the above inequalities. The hardness of an LCP arises from the exponential number of possible choices for $S.$ 
This also demonstrates that although an LCP is ostensibly a continuous optimization problem, it implicitly encodes a problem of combinatorial character.

Results on LCPs concern questions such as existence, uniqueness, and boundedness of solutions, and their stability to changes in the vector $q$, in addition to computation. A typical line of attack has been to characterize classes of matrices $M$ and vectors $q$ for which the $\LCP(M,q)$ has the desired properties. A vast variety of matrix classes have been analyzed; we refer the reader to~\cite{cottle92linear} more on this topic.

\subsection{$\LCP(A+I,-\bfe)$ and its properties}
For a graph $G = (V,E)$ with vertices $V=\{1,\hdots,n\}$ we now study a few properties of $\LCP(G)$, \ie, $\LCP(A+I,-\bfe)$ where $A$ is the adjacency matrix of $G$, $I$ is the $n\times n$ identity matrix, and $\bfe$ is the vector of 1's. We define the \textit{support} of a vector $x \in \Real^n$  as $$\sigma(x) := \{i \in V \mid x_i > 0\}.$$ For $S,S'\subseteq V$, we denote by $N_S(S'):= \{j \in S \mid \exists i \in S' \sthat a_{ij} = 1 \},$ 
the neighbourhood of set $S'$ relative to $S$. For a singleton $S' = \{u\}$ we denote it by $N_S(u)$ and the subscript is dropped if $S =V$. For $x \in \Real^{n}$, we define $\mathcal{C}(x) := (A+I)x$ and denote by $\scn{i}{x}$ the $i^{th}$ component of $\mathcal{C}(x)$. We call $\scn{i}{x}$ the \textit{sum of the closed neighbourhood} of $i$ with respect to $x$. Clearly,
\begin{equation}
 \scn{i}{x} := x_i + \underset{j \in V}\sum a_{ij}x_j = x_i + \underset{j \in N(i)}\sum x_j. \tag{2} \label{eq:sum of closed neighbourhoods} 
\end{equation}

Observe that vector $x$ solves $\LCP(G)$ is equivalent to
\begin{align}
&x \geq 0,  &&\qquad \Leftrightarrow 
\qquad x_i \geq 0, \quad\forall \;i\in\; V, \label{eq:lcpg1} \\
&(A_G + I)x \geq \bfe, &&\qquad \Leftrightarrow \qquad \scn{i}{x} \geq 1, \quad\forall \;i\in\; V, \label{eq:lcpg2} \\
&x \t \big((A+I)x - \bfe\big) = 0. &&\qquad \Leftrightarrow \qquad x_i (\scn{i}{x}-1)=0, \quad\forall \;i\in\; V. \label{eq:lcpg3} 
\end{align}

For the rest of the paper, the constraint $x_i=0{\rm\;or\;}\scn{i}{x}=1$ is called the complementarity constraint for vertex $i$. Note that $\LCP(G)$ may have fractional solutions. For example if $G$ is regular with degree $d$, then $x = \bfe/(d+1)$ solves $\LCP(G).$
We now study a few additional properties of the structure of the $\LCP(A+I,-\bfe)$. For a graph $G$, let $V(G)$ denote the vertex set of $G$, and for a set $S \subseteq V(G)$, let $G_S$ denote the subgraph induced by $S$. For a vector $x$ of size $|V(G)|$, denote by $x_S$ the corresponding subvector of $x$ indexed by vertices in $S$.

\begin{lemma}
\label{lem:basic lcpg results}
Consider a graph $G = (V,E)$ and the $\LCP(G) = \LCP(A+I,-\bfe)$, where $A$ is the adjacency matrix of $G$, $I$ is the $n \times n$ identity matrix, and $\bfe$ is the vector of 1's. Then,

\begin{enumerate}[label=(\alph*)]
\renewcommand{\theenumi}{$(\alph{enumi})$}
\item \label{lem:zero notin SOL}$0 \notin \SOL(G)$,
\item $\mathcal{C}(x) \geq x, \ \forall \; x \in \SOL(G)$,
\item \label{lem:sol subset of [0,1]} $\SOL(G) \subseteq [0,1]^n$,
\item \label{lem:disjoint}If a graph $G$ is a disjoint union of graphs $G_1$ and $G_2$, then, $\SOL(G) = \SOL(G_1)\times\SOL(G_2)$.
\item \label{lem:dominating}For a graph $G$, if $x \in \SOL(G)$, $\sigma(x)$ is a dominating set of $G$.
\item \label{lem:subvector solution}For a graph $G$, if $x \in \SOL(G)$, then $\hat{x}:=x_{\sigma(x)} \in \SOL(G_{\sigma(x)})$ and $\sigma(\hat{x}) = V(G_{\sigma(x)})$.
\end{enumerate}
\end{lemma}
\begin{proof}
\begin{enumerate}[label=(\alph*)]
\item Clearly, $\scn{i}{0} = 0$ for all $i \in V$, which violates (3). Thus $0 \notin \SOL(G).$
\item Let $x \in \SOL(G)$, then $x_j \geq 0,\; \forall \; j \in V$. By definition $\scn{i}{x} = x_i + \sum_{j\in N(i)}x_j$. Hence $\scn{i}{x} \geq x_i$ for all $i \in V$ with equality occurring only when $x_i = 1$.
\item \label{lem:interval} Let $x \in \SOL(G)$ then $x_i \geq 0$. Now suppose $x_i > 1$ for some $i$, then $\scn{i}{x} > 1$ and (4) is violated. Hence $0 \leq x_i \leq 1$.
\item \label{lem:disjoint subgraphs} Let $A$, $A_1$ and $A_2$ be the adjacency matrices of $G$, $G_1$ and $G_2$ respectively. Let $x \in \SOL(G)$ and let for $i = 1,2$, $x^{(i)}$ and $\bfe^{(i)}$ respectively denote the subvectors of $x$ and $\bfe$ indexed by vertices in $G_i$. Observe that since $G$ is a disjoint union of two graphs, $A$ is a block diagonal matrix with $A_1$ and $A_2$ as diagonal blocks.

Since $x \in \SOL(G)$, we have $x \geq 0$, $(A+I)x \geq e$ and $x\t\big((A+I)x - \bfe\big)=0$. This means for $i = 1,2$, $x^{(i)} \geq 0$, $(A+I)x^{(i)} \geq e^{(i)}$ and $x^{(i)\top}\big((A+I)x^{(i)} - \bfe^{(i)}\big)=0$, whereby $x^{(i)} \in \SOL(G_i)$. Conversely, if $x^{(i)} \in \SOL(G_i)$ for $i=1,2$ then $x=(x^{(1)},x^{(2)}) \in \SOL(G).$
 This proves the lemma.
\item If $x \in \SOL(G)$, then by Lemma \ref{lem:basic lcpg results} \ref{lem:interval} we have, $0 \leq x_i \leq 1,\;\forall \;i \in V$. Hence $\bfone_{\sigma(x)} \geq x$ for all $i \in V$, whereby $\scn{i}{\bfone_{\sigma(x)}} \geq \scn{i}{x} \geq 1$. This means that every vertex not in $\sigma(x)$ has at least one neighbour in $\sigma(x)$. This proves that $\sigma(x)$ is a dominating set.
\item Let $x \in \SOL(G)$ and $\hat{x}:=x_{\sigma(x)}$, whereby $\hat{x}_i > 0$ for all $i \in \sigma(x)$. Hence for the graph $G_{\sigma(x)}$, $\sigma(\hat{x}) = V(G_{\sigma(x)}) = \sigma(x)$. Moreover, notice that for a vertex $i$ in $G_\sigma(x)$, the sum of closed neighbourhoods denoted by $\hat{\mathcal{C}}_i(\hat{x}) = \scn{i}{x}$, since all vertices $i \in V(G)\backslash\sigma(x)$ have $x_i=0$. Hence for $i$ in $V(G_\sigma(x))$ we have, $\hat{x}_i \geq 0$, $\hat{\mathcal{C}}_i(\hat{x}) \geq 1$ and $\hat{x}_i(\hat{\mathcal{C}}_i(\hat{x})-1) = x_i(\scn{i}{x}-1) = 0$. Hence $\hat{x} \in \SOL(G_{\sigma(x)})$.
\end{enumerate}
\end{proof}

We now study a property associated with the integer solutions of $\LCP(A+I,-\bfe)$.
\begin{lemma}
\label{lem:integer solutions of lcpg}
For a graph $G= (V,E)$, a vector $x$ is an integral solution of $\LCP(G)$ if and only if it is the characteristic vector of a maximal independent set of $G$.
\end{lemma}

\begin{proof}
From Lemma \ref{lem:basic lcpg results} \ref{lem:sol subset of [0,1]}, we know that integer solution to $\LCP(G)$ is necessarily a binary vector and hence it is the characteristic vector of some set contained in $V$.

Consider such a binary vector $\bfone_S$ for some set $S \subseteq V$. It always satisfies $\bfone_S \geq 0$. We first show that $\bfone_S$ satisfying the complementarity constraint \eqref{eq:lcpg3} is equivalent to $S$ being an independent set. Next, we show that, if $S$ is independent, then $\bfone_S$ satisfying \eqref{eq:lcpg2} is equivalent to $S$ being a maximal independent set. These claims together complete the proof.

First we note that $S\subseteq V$ is an independent set if an only if  the sum $\sum_{i \in S}\sum_{j \in S}a_{ij} = 0$: If $S$ is an independent set then $a_{ij} = 0$ for all $i,j\in S$ and hence  this sum is 0. 
Conversely, if this sum vanishes, then all the terms $a_{ij}$ appearing in it, being non-negative are necessarily zero whereby $S$ is an independent set. Observe that this sum is in fact $\sum_{i \in V}\sum_{j \in V}a_{ij}(\bfone_{S})_i(\bfone_{S})_j = \bfone_S \t  A \bfone_S$. Since $\bfone_S$ is binary, $\bfone_S \t \bfone_S = \bfe \t \bfone_S$ and hence $\bfone_S((A+I)\bfone_S - \bfe) = \bfone_S \t  A \bfone_S$. Hence, 
$$S {\rm \;is \;independent} \qquad \iff \qquad \bfone_S \t ((A+I) \bfone_S-\bfe) = 0 \qquad \iff \qquad \bfone_S\ {\rm satisfies \ } (5).$$
Finally, if $S$ is an independent set, then $\scn{i}{\bfone_S} = 1, \; \forall \;i\in S$. Moreover $\scn{i}{\bfone_S} \geq 1$ for all $i \notin S$ means every vertex not in $S$ has at least one neighbour in $S$. Recall that this is a property of maximal independent sets. Hence,
$${\rm If\;} S {\rm\;is\;independent\ and}\ (A+I)\bfone_S\geq \bfe \iff S {\rm\;is\;a \ maximal\ independent\ set}.$$
This concludes the proof of the lemma.
\end{proof}

As a consequence of Lemma \ref{lem:integer solutions of lcpg}, we have,
$$\alpha(G) = \max\;\{\bfe \t x \mid x\in \{0,1\}^n \cap \SOL(G)\} \qquad {\rm and}\qquad \beta(G) =\min\;\{\bfe \t x \mid x\in \{0,1\}^n \cap \SOL(G)\}.$$ 
The next lemma provides an upper bound on the $\ell_1$-norm of a solution of $\LCP(A+I,-\bfe)$ if the support of the solution contains a maximal independent set.

\begin{lemma}
\label{lem:support contains MIS}
If a maximal independent set of a graph $G$ is contained in the support of a solution to $\LCP(G)$, then the $\ell_1$ norm of the solution is upper bound by the cardinality of the set.
\end{lemma}
\begin{proof}
Let $x$ be a solution of the $\LCP(G)$ such that $\exists \; S \subset \sigma(x)$, and $S$ is a maximal independent set. We have to show $\bfe \t x \leq |S|$. Let $U := \sigma(x) \backslash S$.  Then,  $\forall\; i \in \sigma(x),$
$$ \quad \scn{i}{x} = \sum_{j \in V}a_{ij}x_j + x_i = 1.$$
Summing over $i\in \sigma(x)$ gives,
\[ \bfe \t x + \sum_{i \in \sigma(x)}\sum_{j\in V}a_{ij}x_j = |\sigma(x)| = |S| + |U|.\]
Thus,
\begin{eqnarray}
|S| - \bfe \t x
&\stackrel{(a)}{=}& \sum_{i\in S}\sum_{j\in U} a_{ij} x_j + \sum_{i\in U} \sum_ j a_{ij} x_j - |U|. \nonumber\\
&\stackrel{(b)}{=}& \sum_{j \in U} |N_S(j)|x_j + \sum_{i\in U} \sum_ j a_{ij} x_j - |U|. \nonumber\\
&\stackrel{(c)}{=}& \sum_{j \in U} |N_S(j)|x_j - \sum_{i\in U} x_i. \nonumber
\end{eqnarray}
The equality in $(a)$ follows from splitting the first summation and applying $a_{ij} = 0$ for all $i,j \in S$. To obtain the equality in $(b)$ the order of summation in the first term is interchanged. Equality $(c)$ is obtained by adding constraints $\scn{i}{x} = 1, \; \forall \, i \in U$.

Recall that for a maximal independent set $S$, every vertex not in $S$ has at least 1 neighbour in $S$. Hence $|N_S(x)| \geq 1$ for all $i \in U$. Hence we have,
$$\bfe \t x \leq |S|, \quad \forall\; S \subseteq \sigma(x), \;{\rm such \;that}\; S {\rm\; is\; maximally \;independent}.$$
This proves the lemma.
\end{proof}

Lemma \ref{lem:support contains MIS} describes an upper bound for solutions containing a maximal independent set in their support. If the graph $G$ is a \textit{forest}, i.e., a collection of trees then for every solution of the $\LCP(G)$ there exists a maximal independent set in its support. This is proved later in Lemma \ref{lem:forest existence of MIS}.

The following lemma states a few results regarding $\SOL(G)$ when $G$ belongs to a few specific classes of graphs namely regular graphs, cliques and trees. 

\begin{lemma}
\label{lem:structure of SOL}
\begin{enumerate}[label=(\alph*)]
\item For a complete graph $K_n$ over $n$ vertices, $\SOL(K_n) = \Delta_n := \{x \geq 0 \mid \bfe \t x = 1\}$.
\item \label{lem:tree SOL}For a forest $G$, if $x \in \SOL(G)$ and $\sigma(x) = V(G)$ then $G$ is a disjoint union of $K_1$ or $K_2$.
\item \label{lem:forest SOL}For a forest $G$, if $x \in \SOL(G)$ then $G_{\sigma(x)}$ is a union of $K_1$'s and $K_2$'s.
\item \label{lem:Regular graphs beta bound}For a regular graph $R_{n,d}$ over $n$ vertices with degree $d$, $\beta(R_{n,d}) \geq \frac{n}{d+1}$.

\end{enumerate}

\end{lemma}
\begin{proof}
\begin{enumerate}[label=(\alph*)]
\item For a complete graph $A+I = \bfe \bfe \t$, the matrix of all ones. Let $x \in \SOL(K_n)$, then the complementarity constraint simplifies to $x\t\bfe (\bfe\t x - 1) = 0$. Since $x \geq 0$ and $0 \notin \SOL(K_n)$ by Lemma \ref{lem:basic lcpg results} \ref{lem:zero notin SOL}, $x \t \bfe > 0$. This implies $x \in \Delta_n$. 

Observe that if the graph is $K_n$, $\scn{i}{x} = \bfe \t x$ for all $x \in \Real^n$. Let $x \in \Delta_n$, then $x \geq 0$, $\mathcal{C}(x) = \bfe$, whereby $x \in \SOL(K_n)$. Hence $\SOL(K_n) = \Delta_n$.

\item We first show that if the graph is a tree and there exists a solution to the $\LCP$ with full support, then the tree must be either $K_1$ or $K_2$. Consider a tree $T\neq K_1$ or $K_2$, and let $x \in \SOL(T)$ with $\sigma(x) = V(T)$, i.e., $x_i >0$ for all $i\in V(T)$. Then $\scn{i}{x} = 1$ for all $i$ due to the complementarity constraint. Consider a leaf vertex $i^*$ of $T$ and its neighbour $j^*$. Since $T \neq K_2$, degree of $j^* \geq 2$. Hence $\big{\{}N(i^*)\cup\{i^*\}\big{\}} \subset \big{\{}N(j^*)\cup\{j^*\}\big{\}}$, a strict subset. Hence we have, $\scn{i^*}{x} < \scn{j^*}{x}$ which is a contradiction. This proves the claim.

Now consider a forest $G$, and let $T^{(i)}$ be its $i^{\rm th}$ connected component. Let $x \in \SOL(G)$ such that $\sigma(x) = V(G)$, and let $x^{(i)}$ denote the subvector of $x$ indexed by vertices in $T^{(i)}$. By Lemma \ref{lem:basic lcpg results} \ref{lem:disjoint subgraphs}, we know that $x^{(i)} \in \SOL(T^{(i)})$. Moreover $\sigma(x^{(i)}) = V(T^{(i)})$ for all components $T^{(i)}$ of $G$. Hence $T^{(i)}$ is either $K_1$ or $K_2$, and \ref{lem:tree SOL} stands proven.

\item For a forest $G$, let $x \in \SOL(G)$. From Lemma \ref{lem:basic lcpg results} \ref{lem:subvector solution} we have, $x_{\sigma(x)} \in \SOL(G_{\sigma(x)})$ and $\sigma(x_{\sigma(x)}) = V(G_{\sigma(x)})$. Observe that $G_{\sigma(x)}$ is also a forest. Hence it follows from Lemma \ref{lem:structure of SOL} \ref{lem:tree SOL}, that $G_{\sigma(x)}$ is a union of $K_1$s and $K_2$s.

\item For a regular graph, $(d,\bfe)$ is an eigenvalue-eigenvector pair of the adjacency matrix. Hence $\frac{\bfe}{d+1} \in \SOL(R_{n,d})$ with $\sigma(\frac{\bfe}{d+1}) =V(R_{n,d})$ and $\bfe \t \frac{\bfe}{d+1} = \frac{n}{d+1}$. Using Lemma \ref{lem:support contains MIS} proves \ref{lem:Regular graphs beta bound} since $\beta(R_{n,d})$ is the cardinality of the smallest maximal independent set of $R_{n,d}$.
\end{enumerate}
\end{proof}

\section{Main results}
\label{sec:LPCC formulation}
\subsection{Proof of Theorem \ref{thm:main}}

For a vector of non-negative\footnote{It can be easily shown that for unconstrained $w$, $\alpha_w(G) = \alpha_{w^+}(G_+)$ where $G_+$ is the subgraph of $G$ over vertices with non-negative weights $w^+$. Thus we only consider non-negative weight vectors $w$ for the rest of the paper.} weights $w \in \Real^n$, let,
\begin{equation}
\label{eq:main LPCC weighted}
M_w(G) := \max \;\{ w \t  x \mid x {\rm \;solves\;} \LCP(A + I, -\bfe)\}.
\end{equation}
To reiterate the statement of the theorem -- \textit{For a simple graph $G$,}
\begin{equation*}
\alpha_w(G) = M_w(G).
\end{equation*}

\begin{proof}[of Theorem~\ref{thm:main}]
We prove Theorem \ref{thm:main} by showing inequalities in both directions. From Lemma \ref{lem:integer solutions of lcpg}, for a simple graph $G$, the characteristic vector of every maximal independent set is a solution to the $\LCP(G)$. The maximum weighted independent set $S^*$ being a maximal independent set\footnote{This is true only since $w_i\geq 0$. One can easily construct a graph with unconstrained vertex weights such that the maximum weighted independent set is not a maximal independent set.} gives a feasible vector for the maximization problem \eqref{eq:main LPCC weighted}. Hence,
$$\alpha_w(G) = w \t \bfone_{S^*} \leq M_w(G).$$

We show $\alpha_w(G) \geq M_w(G)$ by induction on the number of vertices $n$ of $G$. For the graph $G_1$ consisting of a single vertex, the adjacency matrix is the scalar 0 and $\SOL(G_1) = \{1\}$. Thus the statement holds for the base case $n = 1$. 

Let us assume the induction hypothesis for all graphs with $n < k$ vertices, i.e., $$\alpha_w(G) \geq M_w(G), \; \forall \; G=(V,E), \;{\rm such\;that\;} |V| < k \aur \forall\;w\geq 0.$$ Let $G^* = (V^*,E^*)$ be a graph with $k$ vertices labelled $V^* = \{1,2,\cdots k\}$. Let $x^* \in \SOL(G^*)$ be the maximizer of \eqref{eq:main LPCC weighted}, i.e., $M_w(G^*) = w \t x^*.$
\\
\\
\noindent \textbf{Case I}: $\sigma(x^*) = V^*$. Let the maximum weighted independent set be $S$, i.e., $\alpha_w(G) = \sum_{i \in S}w_i$. Let $S^c = V^*\backslash S$ be its complement. The complementarity constraint on $x^*$ dictates $(A+I)x^* = \bfe$, i.e., $\forall\; i \in V,$
\begin{equation}
\scn{i}{x^*}=\sum_{j \in V}a_{ij}x^*_j + x^*_i = 1.  \label{eq:temp1} 
\end{equation}
Hence,
\begin{eqnarray}
M_w(G^*) = \sum_{i\in S^c}w_ix^*_i + \sum_{i\in S}w_ix^*_i &\stackrel{(d)}{=}& \sum_{i \in S^c} w_i x^*_i + \sum_{i \in S}w_i  - \sum_{i\in S}\sum_{j\in V} w_i a_{ij} x^*_j, \nonumber\\
&\stackrel{(e)}{=}& \alpha_w(G^*) + \sum_{j \in S^c} w_j x^*_j  -  \sum_{i\in S}\sum_{j \in S^c} w_i a_{ij} x^*_j, \nonumber\\
&=& \alpha_w(G^*) - \sum_{j \in S^c} x^*_j \cdot \left (\sum_{i \in S} w_i a_{ij} - w_j\right ). \label{eq:temp2} 
\end{eqnarray}
Here $(d)$ is obtained by multiplying each equation  \eqref{eq:temp1} by $w_i$ and adding these equations for $i \in S$, and then substituting the resulting resulting expression for $\sum_{i\in S} w_i x_i^*$.
$(e)$ follows from using  $a_{ij} = 0$ for $ i,j \in S$ and $\sum_{i \in S}w_i = \alpha_w(G^*)$. 

We now show an intermediate inequality $ \sum_{i \in S} w_i a_{ij} - w_j \geq 0, \; \forall \;j \in S^c$. To prove this suppose the contrary holds for some $j^* \in S^c$, i.e., $ \sum_{i \in S} w_i a_{ij^*} - w_{j^*} < 0 $. Now consider the set $S' = \{S \backslash N_S(j^*)\}  \cup \{j^*\}$. Clearly $S'$ is an independent set. Moreover, the weight of $S'$ is greater than the weight of $S$ by $w_{j^*} - \sum_{i \in S} w_i a_{ij^*}$, a positive quantity, by assumption. This contradicts that $S$ is a weighted maximum independent set. Hence $ \sum_{i \in S} w_i a_{ij} - w_j \geq 0, \; \forall \;j \in S^c$ whereby, from \eqref{eq:temp2},
$$M_w(G^*) \leq \alpha_w(G^*),$$
as required.
\\
\noindent \textbf{Case II}: $\sigma(x^*) \subset V^*$, a strict subset. Let $x^*_k = 0$ without loss of generality. Let $G_U^*$ be the subgraph of $G^*$ induced by $U := V^*\backslash\{k\}$. Let $y,\widetilde{w}$ be vectors in $\Real^{k-1}$ such that $y_i = x^*_i $ and $ \widetilde{w}_i = w_i$ for all $i \in U$. Clearly $y \geq 0$. Also for all $i \in U$,

$$\scn{i}{y} = \sum_{j \in U} a_{ij}y_j + y_i = \sum_{j \in U} a_{ij}x^*_j + x^*_i = \sum_{j \in V^*} a_{ij}x^*_j + x^*_i = \scn{i}{x^*}\geq 1  .$$
Moreover for all $i \in U$,
$$y_i(\scn{i}{y} - 1) = x^*_i(\scn{i}{x^*} - 1) \stackrel{(f)}{=}0,$$
where $(f)$ follows due to $x^* \in \SOL(G^*)$. Hence $y \in \SOL(G_U^*)$ and we have, $$M_w(G^*) = w \t x^* = \widetilde{w} \t y \leq M_{\widetilde{w}}(G_U^*).$$ 

The inequality above holds since $y$ is a feasible vector for the maximization program defining $M_{\widetilde{w}}(G_U^*)$. Now since $G_U^*$ is a graph with $< k$ vertices, the induction hypothesis dictates that $ M_{\widetilde{w}}(G_U^*) \leq \alpha_w(G_U^*)$. Moreover, since $G_U^*$ is a subgraph of $G^*$, every independent set in $G_U^*$ is an independent set in $G^*$ and hence we have $\alpha_w(G_U^*) \leq \alpha_w(G^*)$. Hence we have,
$$M_w(G^*) \leq \alpha_w(G^*).$$

After considering two exhaustive cases, the inequality $M_w(G^*) \leq \alpha_w(G^*)$ is proved. This concludes the proof of Theorem \ref{thm:main}. 
\end{proof}

The uniform weighted version of Theorem \ref{thm:main} is stated below.

\begin{theorem}
\label{thm:alpha LPCC}
For any simple graph $G$,
$$\alpha(G) = \max\; \{\bfe \t x \mid x\in \SOL(G)\} = \underset{x\in\mathbb{Z}^n}\max\; \{\bfe \t x \mid x\in \SOL(G)\},$$
where $A$ is the adjacency matrix of $G$, $I$ is the identity matrix and $\bfe$ is the vector of 1's in $\Real^n$.
\end{theorem}

Recall that from Lemma \ref{lem:integer solutions of lcpg}, for a graph $G=(V,E)$, with adjacency matrix $A$, the integer solutions of $\LCP(A+I,-\bfe)$ are characteristic vectors of maximal independent sets of $G$, whereby $\alpha(G)$ is the maximum $\ell_1$ norm of binary vectors in $\SOL(G)$. The next section discusses the minimum $\ell_1$ norm of points in $\SOL(G)$.

\subsection{Minimum $\ell_1$ norm solution of $\LCP(G)$ and the independent domination number}
\label{sec:minimum l1 norm of solg}
We now study a few properties of the minimum $\ell_1$ norm of vectors in $\SOL(G)$. We define,
\begin{equation}
\label{eq:define minimum l1 norm}
m(G) := \min \;\{\bfe \t x \mid x\in \SOL(G)\}.
\end{equation}
Interestingly, unlike the maximum $\ell_1$ norm, the minimum $\ell_1$ norm is not necessarily an integer. Recall that for a simple graph G, $\beta(G)$ is the size of the smallest maximal independent set. Hence by Lemma \ref{lem:integer solutions of lcpg} we have,
\begin{equation}
\label{eq:beta gap}
\beta(G) = \min\;\{\bfe \t x \mid x\in \SOL(G)\cap\bin^n\} \geq m(G).
\end{equation}

The inequality above is strict in general and we show that the gap is strict even for bipartite and regular graphs. However equality is guaranteed if the graph is a forest (this is the claim of Theorem~\ref{thm:forest min l1 beta} we prove below).

\begin{lemma}
\label{lem:regular graphs min l1}
For a regular graph $R_{n,d}$ with $n$ vertices and degree $d$, $$m(R_{n,d}) = \frac{n}{d+1}.$$
\end{lemma}
\begin{proof}
Observe that $\frac{n}{d+1}$ is the $\ell_1$ norm of the vector $\frac{\bfe}{d+1}$. Recall that if $A$ is the adjacency matrix of $R_{n,d}$, then $(d,\bfe)$ is an eigenvalue-eigenvector pair and hence $\frac{\bfe}{d+1} \in \SOL(R_{n,d})$. Hence, \begin{equation} \label{eq:reg1}
 m(R_{n,d})\leq \frac{n}{d+1}
.\end{equation}
It can also be verified that $\frac{\bfe}{d+1} \in \{x\mid x\geq 0;\; (A+I)x \leq \bfe\}$. 
Hence,
$$\frac{n}{d+1}\leq\max \;\{\bfe \t x \mid x \geq 0;\; (A + I)x\leq \bfe\}\stackrel{(g)}{\leq}\min \;\{\bfe \t x \mid x \geq 0;\; (A + I)x\geq \bfe\},$$
where $(g)$ follows from LP duality. It is easy to see that $m(R_{n,d}) \geq \min \;\{\bfe \t x \mid x \geq 0;\; (A + I)x\geq \bfe\}$, the LP relaxation of $m(R_{n,d})$ obtained after omitting the complementarity constraint. Hence we have, \begin{equation} \label{eq:reg2}
m(R_{n,d})\geq \frac{n}{d+1}
.\end{equation}
Equations \eqref{eq:reg1} and \eqref{eq:reg2} together prove the lemma.
\end{proof}

Notice that Lemma \ref{lem:regular graphs min l1} along with \eqref{eq:beta gap} gives an alternate proof for Lemma \ref{lem:structure of SOL} \ref{lem:Regular graphs beta bound}. We now look at an example where there is a non-zero gap between $\beta(G)$ and the smallest $\ell_1$ norm of vectors in $\SOL(G)$.

\begin{examplec}
\label{ex:beta gap}
For $k \in \mathbb{N}$, consider the cycle $C_{6k+2}$, with vertex set $V = \{1,2,\hdots, 6k+2\}$ and edges $(i,i+1)$ for $1 \leq i \leq 6k+1$, and $(1,6k+2)$. This graph is a regular bipartite graph with degree $2$. Using Lemma \ref{lem:regular graphs min l1} we have $m(C_{6k+2}) = 2k+1/3$. We now show that $\beta(C_{6k+2}) = 2k+1$. 

Since $\beta(C_{6k+2})$ is an integer lower bounded by $2k+1/3$ (by equation \eqref{eq:beta gap}), it suffices to show that there is a maximal independent set of $C_{6k+2}$ of size $2k+1$. Consider the set $\{1\}\cup\{3i| 1 \leq i \leq 2k\}$. It can be verified that this set is maximally independent with cardinality $2k+1$. Hence we have, $$\beta(C_{6k+2}) = 2k+1 > 2k + 1/3 = m(C_{6n+2}).$$
\end{examplec}

The above example shows that the gap between $\beta$ and $m$ is non-zero even for bipartite and regular graphs.

\subsection{Proof of Theorem \ref{thm:forest min l1 beta}}

Theorem \ref{thm:forest min l1 beta} says that the independent domination number of the graph, $\beta(G)$, is indeed the minimum $\ell_1$ norm of vectors in $\SOL(G)$, provided that the graph $G$ is a forest. Before proving the theorem, we first state and prove an intermediate lemma that examines the structure of the solution set of the $\LCP(A+I,-\bfe)$ for forests.

\begin{lemma}
\label{lem:forest existence of MIS}
For a forest $G=(V,E)$, if $x \in \SOL(G)$, then there exists a maximal independent set of $G$ contained in $\sigma(x)$.
\end{lemma}
\begin{proof}
Let $G = (V,E)$ be a forest and $x \in \Real^{|V|}$. Denote by $G_{\sigma(x)}$ the subgraph of $G$ induced by $\sigma(x)$. By Lemma \ref{lem:structure of SOL} \ref{lem:forest SOL} we know that if $x \in \SOL(G)$, $G_{\sigma(x)}$ is a union of $K_1$'s and $K_2$'s. Let the graph induced by these resulting $K_1$'s be denoted by $G_1=(V_1,\emptyset)$ and the graph induced by the $K_2$'s be denoted by $G_2=(V_2,E_2)$ with edges $E_2 = \{(u_1,u_1'),(u_2,u_2'),\ldots,(u_m,u_m')\}$. Let $U = \{u_i \mid i = 1,2\ldots,m\}$ and $U' = \{u_i' \mid i = 1,2\ldots,m\}$, whereby $V_2= U\cup U'.$  For $J \subseteq \{1,2,\ldots,m\}$ we denote $U_J := \{u_i \mid i \in J\}$ and similarly for $U'_J$.

Let $W = V\backslash \big{\{}V_1 \cup N(V_1)\big{\}}$. And $G_W = (V_W,E_W)$ be the subgraph of $G$ induced by $W$. Observe that $G_W$ is also a forest and $V_W$ is the disjoint union of $U , U'$ and  $W'$, where $W'$ is the set of vertices in $V$ with neighbours only in $V_2$ but not in $V_1$. Notice that $0< x_u < 1, \; \forall\; u \in U \cup U'$ since any such $u$ has exactly one neighbour that is also in $\sigma(x)$. But since $\scn{u}{x} \geq 1, \;\forall\; u \in W',$ every vertex $u \in W'$ has at least two neighbours in $U \cup U'$. This implies that the number of edges in $G_W$ is at least $m + 2|W'|$. Since $G_W$ is a forest, $m + 2|W'| \leq |W'| + |U| + |U'| - 1= |W'| + 2m - 1$ which proves $W'$ has at most $m-1$ vertices.

It is easy to see that if for some $J \subseteq \{1,2,\ldots,m\}$, $U_J \cup U_{J^c}'$ is a maximal independent set of $G_W$, then $\big{\{}V_1 \cup U_J \cup U_{J^c}'\big{\}}$ is a maximal independent set of $G$ contained in $\sigma(x)$. Thus we prove the former by constructing a maximal independent set of $G_W$ which is contained in $V_2=U\cup U'$. All neighbourhoods in the rest of the proof are with respect to the forest $G_W$.

For the forest $G_{W}$, there exists a leaf vertex in $V_2$ since every vertex in $W'$ has degree at least two. Let $u_m'$ be the leaf vertex without loss of generality. Now consider the subgraph $G_W^{(1)}$ of $G_W$ induced by $W^{(1)}:= \big{\{}W'\backslash N_{W'}(u_m)\big{\}} \cup \big{\{}U_{m-1} \cup U_{m-1}'\big{\}}$, where $U_{m-1}:=U\backslash\{u_m\}$ and $U_{m-1}'=U' \backslash \{u_m'\}$. Observe that $G_W^{(1)}$ is now a forest with the property that vertices in $\big{\{}W'\backslash N_{W'}(u_m)\big{\}}$ have degree at least two, whereby there exists a leaf vertex of $G_W^{(1)}$ in $\{U_{m-1} \cup U_{m-1}'\}$. Moreover, by construction, if $M^{(1)}$ is a maximal independent set in $G_W^{(1)}$ then $\{u_m\}\cup M^{(1)}$ is a maximal independent set in $G_W$. Thus it now suffices to find a maximal independent set of $G_W^{(1)}$ contained in $U_{m-1} \cup U_{m-1}'$. We can proceed in a similar manner for $k $ steps (say, $k\leq m-1$)  by picking the neighbour of some leaf node at every step until the induced subgraph of $G_W$ we are left with, $G_W^{(k)}$, is a subgraph of $G_2$. Choosing any maximal independent set of $G_W^{(k)}$ and the previously chosen neighbours of leaf vertices gives a maximal independent set of $G_W$ contained in $V_2$. This concludes the proof.
\end{proof}

We now prove Theorem \ref{thm:forest min l1 beta} which says that -- \textit{For a forest }$G$, $$m(G) = \beta(G).$$

\begin{proof}[of Theorem~\ref{thm:forest min l1 beta}] Let $x^*$ be the minimizer of \eqref{eq:define minimum l1 norm}. 
From Lemma \ref{lem:structure of SOL} \ref{lem:forest SOL}, since $G$ is a forest, $G_{\sigma(x^*)}$ is a union of $K_1$'s and $K_2$'s. Let the graph induced by these $K_1$'s be $G_1=G_1(x^*)$ and the graph induced by $K_2$'s be $G_2=G_2(x^*)$. We have that for all $i$ in $V(G_1(x^*))$, $x^*_i = 1$ and for all edges $(i,j)$ in $G_2(x^*)$, since $\scn{i}{x^*}=\scn{j}{x^*}=1 $ we must have $0<x^*_i<1$, $0<x^*_j<1$ and $x^*_i+x^*_j=1$. Hence $m(G)=\bfe \t x^* = |V(G_1(x^*))| + \half |V(G_2(x^*))|$ 

But from Lemma \ref{lem:forest existence of MIS}, we know that there exists a maximal independent set of $G$ contained in $\sigma(x^*)$ such that its size is $|V(G_1(x^*))| + \half |V(G_2(x^*))|$, whereby $m(G)\geq \beta(G)$. This along with \eqref{eq:beta gap} proves the theorem.
\end{proof}

\section{Applications} \label{sec:applications} 

Theorem \ref{thm:main} expresses the independence number as a linear program with complementarity constraints (LPCC). In this section we use methods from the theory of LCPs to provide an upper bound for the independence number and also give a sufficient condition for a graph to be a \textit{well-covered}.

\subsection{A strengthening of the Lov\'asz theta}

The Lov\'{a}sz theta \cite{lovasz1991cones}, denoted $\vartheta(G)$, for a graph $G$ is a polynomially computable upper bound on the independence number obtained via semidefinite programming. The bound is tight for perfect graphs. In this section we introduce a new quantity $\vartheta^*(G)$  that improves on $\vartheta(G).$ 
We first briefly introduce semidefinite programming, then describe the Lov\'{a}sz theta bound. To achieve a strengthening of the bound, we first characterize the independence number using a more compact 0-1 ILP, and then obtain a variant of the Lov\'{a}sz theta using semidefinite relaxations of the ILP.

For a graph $G = (V,E)$ the traditional 0-1 ILP for the independence number $\alpha(G)$ maximizes the $\ell_1$ norm of characteristic vectors of all independent sets in a graph.
\begin{equation}
\label{eq:edge-ilp}
\tag{edge-ILP}
\begin{aligned}
\alpha(G) = \underset{x \in \bin^n}{\text{max}} & \bfe \t x,\\
\text{s.t.} & x_i + x_j \leq 1, \quad \forall\;(i,j) \in E
.\end{aligned}
\end{equation}

It is easy to see that characteristic vectors of all independent sets including the empty set are feasible for \eqref{eq:edge-ilp}. The convex hull of these characteristic vectors is called the stable set polytope of $G$ and is denoted by $\STAB(G)$. 

Semidefinite programs~(SDP) are convex optimization problems which generalize linear programs and are efficiently solvable. They have been studied in detail \cite{boyd04convex, vandenberghe1996semidefinite} and often result in tighter bounds for ILPs than those obtained via LP relaxation. Lov\'{a}sz \cite{lovasz1991cones} provided a semidefinite relaxation for $\alpha(G)$. We now briefly describe semidefinite programs in their most general form and then state the SDP relaxation by Lov\'{a}sz.

Let $\mathcal{S}^n$ be the cone of symmetric positive semidefinite matrices in $\Real^{n \times n}$. For $n\times n$ matrices $X=[x_{ij}]$ and $Y=[y_{ij}]$, their dot product is defined as $\big<X ,Y\big> := \sum_{i=1}^n\sum_{j=1}^n x_{ij}y_{ij}$. A semidefinite program in its most general form is the following convex optimization problem,

\begin{equation*}
\label{eq:SDP}
\tag{SDP}
\begin{aligned}
 \underset{X \in \mathcal{S}^n}{\text{max}}
  &  \big<C ,X\big>\\
  \text{s.t.}
  & \big<A_i \t ,X\big> = b_i\quad i = 1,2,\ldots m,
 \end{aligned}
\end{equation*} where $C$ and $A_i$ are symmetric matrices in $\Real^{n \times n}$, and $b_i$ are scalars. Given a vector $x \in \Real^n$ and matrix $W \in \Real^{n\times n}$, let ${\rm L}(x,W) := \begin{bmatrix} 1 & x \t \\x & W\end{bmatrix}$ and $\diag(W):=(w_{11},w_{22},\cdots,w_{nn})$, the vector of diagonal elements of $W$. The Lov\'{a}sz theta is given by the semidefinite program below,
\begin{equation*}
\label{eq:SDP Lovasz theta}
\tag{$\vartheta$-SDP}
\begin{aligned}
\vartheta(G) =  & \underset{X \in \mathcal{S}^n}{\max}
 & &  \big<\bfe\bfe \t ,X\big>\\
 & \text{s.t.}
 & & \big<I ,X\big> = 1,\\
 &&& X_{ij} = 0, \; (i,j) \in E(G),
 \end{aligned}
\end{equation*}
where $I$ is the $n\times n$ identity matrix. It can be shown (see Chapter 3 in \cite{gvozdenovic2008approximating}) that $\vartheta(G)$ is also the maximum $\ell_1$ norm over the following set,
\begin{equation}
\label{eq:theta body}
\TB(G) := \Big{\{x} \mid \exists\;\; W \in \mathcal{S}^n ,\; \diag(W) = x,\; {\rm L}(x,W) \succeq 0, \;w_{ij} = 0, \;\forall \;(i,j)\in E \Big{\}}.
\end{equation}

This set is called the \textit{theta body} of a graph $G$. It is well known that $\TB(G)$ is convex and that $\STAB(G) \subseteq \TB(G)$ whereby $\alpha(G)\leq \vartheta(G)$. Variants of $\vartheta(G)$ have been explored (\cite[Chapter 3]{gvozdenovic2008approximating} and \cite{dukanovic2007semidefinite}) where additional constraints are added to the semidefinite program to obtain stronger upper bounds than $\vartheta(G)$ via SDP relaxations. One such modified theta body $\TB'(G)$ is obtained by adding constraints $X \geq 0$ to \eqref{eq:SDP Lovasz theta} (this constraint corresponds to $W \geq 0$ in \eqref{eq:theta body} \cite{gvozdenovic2008approximating}). The corresponding Lov\'{a}sz theta variant is $\vartheta'(G)$. One can easily show that $\STAB(G) \subseteq \TB'(G) \subseteq \TB(G)$ and hence $\alpha(G) \leq \vartheta'(G) \leq \vartheta(G)$.

\subsubsection{Compact 0-1 ILP for independence number}

We now look at an new ILP formulation for the independence number of a graph $G$ using properties of $\LCP(G)$. The formulation maximizes  the $\ell_1$ norm over maximal independent sets only unlike \eqref{eq:edge-ilp} where all independent sets are considered. For this we characterize the maximal independent sets of $G$ using linear inequalities. Recall that maximal independent sets of $G$ are integer solutions of $\LCP(G)$. An $\LCP(M,q)$ with a bounded solution set can be characterized by the following mixed integer linear program~(MILP) \cite{kojima2002some}. Specifically, if one knows a priori bounds $r,r'$ on $\norm{x}_\infty$ and $\norm{Mx+q}_\infty$, respectively, for all $x $ in the solution set of $ \LCP(M,q)$, this LCP is equivalent to:
\begin{align*}
\text{Find} x \in \Real^n \ \sthat \exists \ z \in \bin^n \sthat \quad
& (1)\; 0 \leq x \leq r z,\\
& (2)\; 0 \leq Mx + q \leq r'(\bfe - z). \tag*{MILP(M,q)}
\end{align*}
The next lemma modifies this MILP characterization for $\LCP(G)$.

\begin{lemma}
\label{lem:polytope for MIS}
For a graph $G=(V,E)$, a vector $x \in \bin^{|V|}$ solves $\LCP(G)$ if and only if,
$$0 \leq (A + I)x - \bfe \leq (D - I)(\bfe - x),$$
where $A$ is the adjacency matrix of the graph, $D$ is the diagonal matrix of degrees of $G$, $I$ is the $n \times n$ identity matrix and $\bfe$ is the vector of 1's.  
\end{lemma}

\begin{proof}
Observe that the inequalities in the statement of the lemma are in fact, $$0\leq \scn{i}{x} -1 \leq (d_i - 1)(1-x_i),$$ 
for all $i\in V.$
Let $x \in \bin^{|V|}$ satisfy the inequalities in the statement of the lemma. Clearly $x_i \geq 0$ and $\scn{i}{x} \geq 1$. Moreover since $x_i \in \bin$, $x_i > 0$ implies $x_i=1$ whereby the second inequality above forces $\scn{i}{x}=1$. Hence we can conclude that $x_i(\scn{i}{x} - 1)=0$. This implies $x$ solves the $\LCP(G)$. 

For the only if part of the lemma, let $\bfone_S \in \SOL(G)$ implying $S$ is a maximal independent set of $G$ by Lemma \ref{lem:integer solutions of lcpg}. By Lemma \ref{lem:basic lcpg results} \ref{lem:dominating} $S$ is a dominating set whereby $\mathcal{C}(\bfone_S) \geq \bfe$. Moreover since $S$ is a maximal independent set, it can be easily seen $\scn{i}{\bfone_S} -1\leq (d_i -1)$ for all $i$.
This proves the lemma.
\end{proof}

For the convenience of notation, let $${\rm MAXIS}(G) := \{x \in \Real^{|V|} \mid 0 \leq x \leq \bfe,\ 0 \leq (A + I)x - \bfe \leq (D - I)(\bfe - x)\}$$
and \[{\rm MAXIS}_{\mathbb{Z}}(G):= {\rm MAXIS}(G)\cap \mathbb{Z}^{|V|}.\]

 Binary vectors in ${\rm MAXIS}_{\mathbb{Z}}(G)$ are characteristic vectors of maximal independent sets of the graph and vice versa from Lemma \ref{lem:integer solutions of lcpg} and \ref{lem:polytope for MIS}. After characterizing maximal independent sets by linear inequalities, the following theorem states the 0-1 ILP formulation for computing the independence number.

\begin{theorem}
\label{thm:ILP for alpha}
For a graph $G$ with adjacency matrix $A$,
\begin{equation}
\label{eq:ILP for alpha1}
\tag{ILP*}
\begin{aligned}
\alpha(G) = \underset{x \in \bin^n}{\max} & \bfe \t x,\\
{\rm s.t.} \quad & 0 \leq (A + I)x - \bfe \leq (D - I)(\bfe - x),
\end{aligned}
\end{equation}
where $D$ is the diagonal matrix of the degrees of the vertices, $I$ is the $n \times n$ identity matrix and $\bfe$ is the vector of 1's.
\end{theorem}

Note that the above 0-1 ILP for $\alpha(G)$ is efficient since it has $\mathcal{O}(n)$ constraints (independent of the number of edges) as compared to $\mathcal{O}(|E|)$ constraints in \eqref{eq:edge-ilp}. The above ILP is now convexified using lift-and-project method to obtain an upper bound on the independence number via SDP relaxations.

\subsubsection{Lov\'asz Schrijver relaxations}

Given a binary lattice, the Lovasz Schrijver relaxation (Section 10.3 \cite{conforti2014integer}) is a \textit{lift and project} method to approximate its convex hull. Let $P = \{x \mid Fx \geq b\}$ for some matrix $F$ and vector $b$ and $P_{\mathbb{Z}}:=\bin^n\cap P$. First the linear inequalities $Fx \geq b$ are converted to second order inequalities by performing the following multiplications for $i = 1,2,\cdots,n$,
\[x_i (Fx - b) \geq 0  \qquad {\rm and} \qquad (1 - x_i) (Fx - b) \geq 0.\]
The second order terms in these inequalities are then replaced by new variables resulting in linear inequalities in a higher dimensional space. This step is called the \textit{lifting} step. 

In the second order inequalities above, replacing $x_i^2$ by $x_i$, and $x_ix_j$ by $w_{ij}$ gives linear inequalities in a $\Real^{n+\frac{n(n-1)}{2}}$ dimensional space from the lifting step. The linearized inequalities are denoted by $\big<\widetilde{F}, {\rm L}(x,W)\big> \geq \widetilde{b}$ and the resulting \textit{linearized polytope} by $\widetilde{P} = \{(x,W) \mid \big<\widetilde{F}, {\rm L}(x,W) \big>\geq\widetilde{b}\}$ where $W = [w_{ij}]$ is a symmetric matrix. Now consider the two sets derived from $P$,
\[M_+(P) = \{(x,W)\in \widetilde{P} \mid \diag(W) = x,\ {\rm L}(x,W) \in \mathcal{S}^{n+1} \} \aur N_+(P) = \{x \mid \exists\; W \in \mathcal{S}^n,\ (x,W) \in M_+(P)\}.\]
It is easy to see that $M_+(P)$ is convex and $N_+(P)$ being the projection of $M_+(P)$ on the space of $x$ is also convex. This step is hence called the \textit{projection} step. One can show the following (Section 10.3 \cite{conforti2014integer}),
\begin{equation}
\label{eq:compare N+}
\conv(P_{\mathbb{Z}}) \subseteq N_+(P) \subseteq P
.\end{equation} 

Let $\TB^*(G) := N_+({\rm MAXIS}(G))$ be obtained using the Lovasz Schrijver relaxation on \eqref{eq:ILP for alpha1}, i.e.,
$$\TB^*(G) = \{x \mid  \exists\;W \in \mathcal{S}^n,\ \diag(W)=x,\ {\rm L}(x,W)\succeq 0,\ (x,W) \in P^*\},$$ where $P^*$ is the \textit{linearized polytope} derived from MAXIS($G$) given by the inequalities,
\begin{align}
x_i(x_j) \geq 0, &\implies w_{ij} \geq 0, \label{eq:positive W}\\
\left.\begin{aligned}
        (1-x_i)(x_j) \geq 0, \\
        x_i(1 - x_j) \geq 0,
       \end{aligned}
 \right\}
  &\implies  x_i \geq w_{ij}, \\
(1-x_i)(1 - x_j) \geq 0,\ &\implies w_{ij} + 1 \geq x_i + x_j, \\
x_i(\scn{j}{x}-1) \geq 0,\ &\implies w_{ij} - x_i + \sum_{k \in V} w_{ik}a_{jk} \geq 0, \\
(1-x_i)(\scn{j}{x}-1) \geq 0,\ &\implies x_j + x_i - w_{ij} - 1 +  \sum_{k \in V} a_{jk}(x_{k}-w_{ik}) \geq 0, \\
x_i(d_j(1 -x_j) - \scn{j}{x} + 1) \geq 0,\ &\implies (d_j+1)(x_i - w_{ij}) - \sum_{k \in V}a_{ik}w_{ik} \geq 0, \label{eq:edge W_ij}\\
(1-x_i)(d_j(1 - x_j) - \scn{j}{x} + 1) \geq 0,\ &\implies (d_j+1)(1 + w_{ij} - x_j - x_i) + \sum_{k \in V}a_{ik}(w_{ik}-x_k) \geq 0, \label{eq:21}
\end{align}
for all $i$, $j$ in $V$. That is, $P^*=\{(x,W) \mid \diag(W) = x,\ \eqref{eq:positive W}-\eqref{eq:21} \ {\rm hold}\}. $

 Let $\vartheta^*(G)$ be the maximum $\ell_1$ norm on $\TB^*(G)$. We show in the following theorem that this Lov\'{a}sz theta variant is a better upper bound for $\alpha(G)$ than $\vartheta'(G)$.

\begin{theorem}
For a graph $G$,
$$\alpha(G) \leq \vartheta^*(G) \leq \vartheta'(G).$$
\end{theorem}

\begin{proof}
It suffices to show that $\conv({\rm MAXIS}_{\mathbb{Z}}(G)) \subseteq \TB^*(G) \subseteq \TB'(G)$ since the quantities compared in the statement of the theorem are the maximum $\ell_1$ norms on these sets.

From equation \eqref{eq:compare N+}, we have, $$\conv({\rm MAXIS}_{\mathbb{Z}}(G)) \subseteq N_+({\rm MAXIS}(G)) := \TB^*(G).$$
Now consider $x \in \TB^*(G)$, and a $W$ such that $(x,W) \in M_+({\rm MAXIS}(G))$, \ie, $\diag(W) = x$, ${\rm L}(x,W) \succeq 0$ and $(x,W) \in P^*$.

Now, substituting $j =i$ in \eqref{eq:edge W_ij} gives $\sum_{k \in V}a_{ik}w_{ik} \leq 0,$ since $w_{ii} = x_i$. Also, from \eqref{eq:positive W} have $w_{ik} \geq 0$. This means, $$w_{ik} = 0,\qquad \forall \; (i,k) \in E$$ Thus $x \in \TB(G)$ with $W$ having all entries non-negative whereby $x \in \TB'(G)$ implying $$\TB^*(G) \subseteq \TB'(G).$$
This proves that $\vartheta^*(G)$ is a stronger upper bound of $\alpha(G)$ than $\vartheta'(G)$ and hence $\vartheta(G)$.
\end{proof}

A similar convexification of $\STAB(G)$ has been studied previously under the name $N_+(\FRAC(G))$ \cite{conforti2014integer}, where $\FRAC(G)$ is the LP relaxation of $\STAB(G)$. However it is computationally more expensive since the linearized polytope has $\mathcal{O}(n^3)$ constraints in the SDP as against $\mathcal{O}(n^2)$ for $\TB^*(G)$. It remains an open problem to compare $N_+(\FRAC(G))$ and $\TB^*(G)$.

\subsubsection{Simulation results comparing Lov\'asz theta variants}

Denote by $\vartheta_{\FRAC}(G)$ the maximum $\ell_1$ norm over the convex set $N_+(\FRAC(G))$. We computed $\vartheta_{\FRAC}(G),$ $  \vartheta^*(G),\vartheta'(G)$ and $\vartheta(G)$ for Erd\H{o}s–R\'enyi graphs $\mathcal{G}_{(n,p)}$; we were able to perform these simulations only for $n\leq 25$. In this model, a random graph over $n$ vertices is chosen such that each edge is included in the graph with probability $p$ independent from every other edge. Table \ref{tab:variants comparison} shows the values of $\alpha(G)$, $\vartheta_{\FRAC}(G)$, $\vartheta^*(G)$, $\vartheta'(G)$ and $\vartheta(G)$ computed for twelve randomly chosen graphs against the parameters of the distribution from which they were drawn. 
The simulations verify that indeed $\vartheta^*(G) \leq \vartheta'(G).$ 
They also
indicate  $\vartheta_{\FRAC}(G) \geq \vartheta^*(G)$. However, no comparison is known between the sets $N_+(\FRAC(G))$ and $\TB^*(G)$, and hence no conclusions can be drawn regarding the ordering of $\vartheta_{\FRAC}(G)$ and $\vartheta^*(G)$ in general. The simulations were performed using the MATLAB-SDPT3 package \cite{toh1999sdpt3, tutuncu2003solving}.

\begin{table}[h]
\centering
\caption{Comparison of variants of Lov\'asz theta. $G \sim \mathcal{G}_{(n,p)}$.}
\label{tab:variants comparison}
\begin{tabular}{|c|c|c|c|c|c|}
\hline
$(n,p)$ & $\alpha(G)$ & $\vartheta_{\FRAC}(G)$ & $\vartheta^*(G)$ & $\vartheta'(G)$ & $\vartheta(G)$ \\ \hline
(15,0.2)   & 7          & 7.0000     & 7.0187    & 7.0669    & 7.0771    \\ \hline
(15,0.4)   & 7          & 7.0006     & 7.3052    & 7.3153    & 7.3402    \\ \hline
(15,0.6)   & 6          & 6.0290     & 6.0723    & 6.0993    & 6.1320    \\ \hline
(15,0.8)   & 5          & 5.2650     & 5.5226    & 5.5350    & 5.5768    \\ \hline
(20,0.2)   & 9          & 9.1798    & 9.3323    & 9.3496    & 9.3863     \\ \hline
(20,0.4)   & 8          & 8.6916     & 8.9087     & 8.9206    & 8.9459    \\ \hline
(20,0.6)   & 8          & 8.4051    & 8.5759    & 8.6100    & 8.6306    \\ \hline
(20,0.8)   & 9          & 9.0412     & 9.3934     & 9.4858    & 9.4871    \\ \hline
(25,0.2)   & 11         & 11.2714    & 11.6478    & 11.7194    & 11.7337    \\ \hline
(25,0.4)   & 11         & 11.2339    & 11.3933    & 11.3997    & 11.4134    \\ \hline
(25,0.6)   & 10         & 10.2903    & 10.3012    & 10.3120     & 10.3457    \\ \hline
(25,0.8)   & 11         & 11.5396    & 11.7890    & 11.8351    & 11.9137    \\ \hline
\end{tabular}

\end{table}

\subsection{Well covered graphs}
Recall that a graph is well-covered~\cite{plummer1993well} if all its maximal independent sets are of the same cardinality. A sufficient condition for a graph $G$ to be well-covered is that $\bfe \t x$ is constant for all $x \in \SOL(G)$ since this means that all integral solutions of the $\LCP(G)$ have the same $\ell_1$ norm, i.e., the cardinality of all maximal independent sets of $G$ is equal. In general this is not a necessary condition for the graph to well covered due to the gap in $m(G)$ and $\beta(G)$ -- an example is the cycle $C_5$ which is well-covered but has $\ell_1$ norm $=2$ for all integral solutions (all maximal independent sets are of size two). However since $C_5$ is also regular with degree two, $\LCP(C_5)$ has a fractional solution $\frac{\bfe}{3}$, with 
$\ell_1$ norm $=\frac{5}{3}$. 

However, if the graph is a forest, we have the following theorem.

\begin{theorem}
A forest $G$ is well-covered if and only if $\bfe \t x$ is constant for  $x \in \SOL(G).$ Moreover, if $G$ has no isolated vertices  $\bfe\t x =\half |V(G)|$ for all $x\in \SOL(G)$.
\end{theorem} 
\begin{proof}
The if part of the theorem is easy and was discussion in the previous paragraph. For the only if part of the theorem, let $G$ be a well-covered forest, which implies $\alpha(G) = \beta(G)$. From Theorem \ref{thm:forest min l1 beta} we know that $\beta(G) = m(G)$ and hence we have, $$\max \{\bfe \t x \mid x \in \SOL(G)\} = \min \{\bfe \t x \mid x \in \SOL(G)\}.$$ This means the $\ell_1$ norm over the solution set of $\LCP(G)$ is a constant for a forest $G$ to be well-covered.

A forest without isolated vertices is a bipartite graph. Both the partite sets of a bipartite graphs are maximal independent sets. Hence a well-covered bipartite graph  $G$ is necessarily a balanced bipartite graph and the cardinality of every maximal independent set is $\frac{|V(G)|}{2}$. This proves the theorem.
\end{proof}
We remark that well-covered graphs $G$ for which each maximal independent set is of size $\half|V(G)|$ are known as \textit{very well-covered graphs}~\cite{plummer1993well}.

An $\LCP(M,q)$ is said to be \textit{w-unique} if the vector $w = Mx + q$ is invariant over the solution set of the LCP. We give a further sufficient condition for a graph to be well-covered based on the \textit{w-uniqueness} of $\LCP(G)$. 
\begin{lemma}
If $\LCP(G)$ is w-unique, then $\bfe \t x$ is constant for all $x \in \SOL(G)$, whereby $G$ is well-covered.
\end{lemma}
\begin{proof}
Let $\LCP(G)$ be \textit{w-unique}. Let $x^1$ and $x^2$ be solutions of $\LCP(G)$.  Then by the complementarity constraint, $\bfe \t  x^1 - (x^1)\t (A + I)x^1=0$ and $\bfe \t  x^2-(x^2)\t (A + I)x^2 =0 $. Moreover, if $\LCP(G)$ is indeed \textit{w-unique}, then we have $  (A + I)x^1 -\bfe= (A + I)x^2 -\bfe $. Consequently, $(x^1)\t (A + I)x^2 =\bfe\t x^1$ and $(x^2)\t (A+I)x^1=\bfe\t x^2$, implying that $\bfe\t x^1=\bfe\t x^2$. This proves the sufficiency of \textit{w-uniqueness} for $\bfe \t x$ is constant for all $x \in \SOL(G)$.
\end{proof}

\subsection{Complexity of LPCCs}

Theorem \ref{thm:alpha LPCC} characterizes the independence number problem as an LPCC. This reduction proves that solving a bounded LPCC with binary inputs is at least as hard as computing the independence number in a simple undirected graph. Hence solving LPCCs are NP-hard in general. The theorem below uses inapproximability results known for the independence number to comment on the inapproximability of LPCCs. Given data vectors $c,d,b,q \in \Real^n$ and matrices $B,C,N,M \in \Real^{n\times n}$, consider an instance of an LPCC as defined in Section~\ref{sec:contributions}.

\begin{theorem}
\label{thm:LPCC complexity}
For any $\epsilon > 0$, unless $P = NP$, there is no polynomial time algorithm to approximate a bounded binary LPCC over $n$ variables within a factor of $n^{1-\epsilon}$ of its optimal value.
\end{theorem}

\begin{proof}
Suppose that there exists a polynomial time algorithm to approximate a general LPCC given to obtain a value larger than $n^{1-\epsilon}{\rm OPT}$ where ${\rm OPT}$ is the optimal value of this LPCC. Consider the LPCC with a single variable, i.e., $x\in \Real^n$ and let $c=\bfe, q=-\bfe$ and $M = A+I$ for some graph $G$, where $A$ is the adjacency matrix of $G$, and all other parameters taken as $0$. Then we have ${\rm OPT }=\alpha(G)$. This means that the assumed algorithm would give an approximation of the independence number of the graph within the interval $[n^{1-\epsilon}\alpha(G),\alpha(G)]$ in polynomial time for any simple undirected graph. By the theorem of H{\aa}stad \cite{haastad1996clique}, this means $P=NP$. Hence under the assumption $P \neq NP$ there exists no such algorithm and the statement of the theorem stands proven.
\end{proof}

\section{Conclusion}
\label{sec:conc}
In this paper, a continuous optimization formulation for the weighted independence number of a graph is studied. The weighted independence number is expressed as a linear program with complementarity constraints, specifically, the maximization of a linear objective function over the solution set of a linear complementarity problem~(LCP). 

We show that the maximum $\ell_1$ norm over all solutions of $\LCP(A+I,-\bfe)$,  where  for a graph $G$, $A$ is the adjacency matrix, $I$ is the identity matrix and $\bfe$ is the vector of ones in $\Real^{|V|}$, is the independence number of $G$. Integral solutions of this LCP correspond to characteristic vectors of maximal independent sets of the graph. The minimum $\ell_1$ norm on the other hand is a lower bound of the independent domination number of the graph. Although the bound is not tight in general, we show that the bound is tight if the graph is a forest.

We then studied a few applications. The LCP provides a new ILP formulation for the independence number and for the independent domination number of a graph which is more efficient than the previously known edge-based ILP. Performing semidefinite programming relaxations on this ILP, gives a stronger variant of the Lov\'asz theta. 
The characterization indicates a new sufficient condition for a graph to be well-covered, i.e., a graph with all maximal independent sets of the same cardinality. Moreover, this condition is shown to be necessary in the case of a forest.
It is also inferred that solving LPCCs remains not only hard to solve but also hard to approximate within a factor $n^{1-\epsilon}$ of the solution for any $\epsilon > 0$ even if the problem data is binary.

\bibliographystyle{plainini}
\bibliography{ref}

\begin{thebibliography}{10}

\bibitem{boyd04convex}
S.~Boyd and L.~Vandenberghe.
\newblock {\em Convex Optimization}.
\newblock Cambridge University Press, New York, NY, USA, 2004.

\bibitem{chung1989np}
S.-J. Chung.
\newblock Np-completeness of the linear complementarity problem.
\newblock {\em Journal of Optimization Theory and Applications},
  60(3):393--399, 1989.

\bibitem{conforti2014integer}
M.~Conforti, G.~Cornu{\'e}jols, and G.~Zambelli.
\newblock {\em Integer programming}, volume 271.
\newblock Springer, 2014.

\bibitem{cottle92linear}
R.~W. {Cottle}, J.-S. {Pang}, and R.~E. {Stone}.
\newblock {\em The Linear Complementarity Problem}.
\newblock Academic Press, Inc., Boston, MA, 1992.

\bibitem{dukanovic2007semidefinite}
I.~Dukanovic and F.~Rendl.
\newblock Semidefinite programming relaxations for graph coloring and maximal
  clique problems.
\newblock {\em Mathematical Programming}, 109(2-3):345--365, 2007.

\bibitem{grotschel1986relaxations}
M.~Gr{\"o}tschel, L.~Lov{\'a}sz, and A.~Schrijver.
\newblock Relaxations of vertex packing.
\newblock {\em Journal of Combinatorial Theory, Series B}, 40(3):330--343,
  1986.

\bibitem{gvozdenovic2008approximating}
N.~Gvozdenovic.
\newblock {\em Approximating the stability number and the chromatic number of a
  graph}.
\newblock 2008.

\bibitem{harant2000some}
J.~Harant.
\newblock Some news about the independence number of a graph.
\newblock {\em Discussiones Mathematicae Graph Theory}, 20(1):71--79, 2000.

\bibitem{harant1999dominating}
J.~Harant, A.~Pruchnewski, and M.~Voigt.
\newblock On dominating sets and independent sets of graphs.
\newblock {\em Combinatorics, Probability and Computing}, 8(06):547--553, 1999.

\bibitem{haastad1996clique}
J.~H{\aa}stad.
\newblock Clique is hard to approximate within $n^{1-\epsilon}$.
\newblock In {\em Foundations of Computer Science, 1996. Proceedings., 37th
  Annual Symposium on}, pages 627--636. IEEE, 1996.

\bibitem{hu2012lpcc}
J.~Hu, J.~E. Mitchell, and J.-S. Pang.
\newblock An {LPCC} approach to nonconvex quadratic programs.
\newblock {\em Mathematical programming}, 133(1-2):243--277, 2012.

\bibitem{hu2012linear}
J.~Hu, J.~E. Mitchell, J.-S. Pang, and B.~Yu.
\newblock On linear programs with linear complementarity constraints.
\newblock {\em Journal of Global Optimization}, 53(1):29--51, 2012.

\bibitem{kojima2002some}
M.~Kojima and L.~Tun{\c{c}}el.
\newblock Some fundamental properties of successive convex relaxation methods
  on lcp and related problems.
\newblock {\em Journal of Global Optimization}, 24(3):333--348, 2002.

\bibitem{kulkarni09refinement}
A.~A. Kulkarni and U.~V. Shanbhag.
\newblock On the variational equilibrium as a refinement of the generalized
  {N}ash equilibrium.
\newblock {\em Automatica}, 48(1):45--55, 2012.

\bibitem{lovasz1991cones}
L.~Lov{\'a}sz and A.~Schrijver.
\newblock Cones of matrices and set-functions and 0-1 optimization.
\newblock {\em SIAM Journal on Optimization}, 1(2):166--190, 1991.

\bibitem{minty1980maximal}
G.~J. Minty.
\newblock On maximal independent sets of vertices in claw-free graphs.
\newblock {\em Journal of Combinatorial Theory, Series B}, 28(3):284--304,
  1980.

\bibitem{motzkin1965maxima}
T.~S. Motzkin and E.~G. Straus.
\newblock Maxima for graphs and a new proof of a theorem of {T}ur{\'a}n.
\newblock {\em Canad. J. Math}, 17(4):533--540, 1965.

\bibitem{murty1988linear}
K.~G. Murty and F.-T. Yu.
\newblock {\em Linear complementarity, linear and nonlinear programming}.
\newblock Citeseer, 1988.

\bibitem{nash51noncooperative}
J.~Nash.
\newblock {Non-cooperative} games.
\newblock {\em The Annals of Mathematics}, 54(2):286--295, September 1951.

\bibitem{plummer1993well}
M.~D. Plummer.
\newblock Well-covered graphs: a survey.
\newblock {\em Quaestiones Mathematicae}, 16(3):253--287, 1993.

\bibitem{toh1999sdpt3}
K.-C. Toh, M.~J. Todd, and R.~H. T{\"u}t{\"u}nc{\"u}.
\newblock {SDPT3}—a matlab software package for semidefinite programming,
  version 1.3.
\newblock {\em Optimization methods and software}, 11(1-4):545--581, 1999.

\bibitem{tutuncu2003solving}
R.~H. T{\"u}t{\"u}nc{\"u}, K.~C. Toh, and M.~J. Todd.
\newblock Solving semidefinite-quadratic-linear programs using sdpt3.
\newblock {\em Mathematical programming}, 95(2):189--217, 2003.

\bibitem{vandenberghe1996semidefinite}
L.~Vandenberghe and S.~Boyd.
\newblock Semidefinite programming.
\newblock {\em SIAM review}, 38(1):49--95, 1996.

\end{thebibliography}

\end{document}